\documentclass{amsart}
\usepackage{amssymb}
\usepackage{amsmath}
\usepackage{amsfonts}

\newtheorem{theorem}{Theorem}
\theoremstyle{plain}

\newtheorem{corollary}{Corollary}

\newtheorem{lemma}{Lemma}
\newtheorem{notation}{Notation}

\newtheorem{remark}{Remark}

\numberwithin{equation}{section}
\input{tcilatex}

\begin{document}
\title[Direct and Inverse Problems for the Heat Equation]{Direct and Inverse
Problems for the Heat Equation with a Dynamic type Boundary Condition }
\author{Nazim B. Kerimov$^{\ast }$, Mansur I. Ismailov$^{\ast \ast }$ }
\address{$^{\ast }$Department of Mathematics, Mersin University, Mersin
33343, Turkey $^{\ast \ast }$Department of Mathematics, Gebze Institute of
Technology, Gebze-Kocaeli \ 41400, Turkey}
\email{nazimkerimov@yahoo.com; mismailov@gyte.edu.tr }
\date{March, 02, 2013}
\dedicatory{}

\begin{abstract}
This paper considers the initial-boundary value problem for the heat
equation with a dynamic type boundary condition. Under some regularity,
consistency and orthogonality conditions, the existence, uniqueness and
continuous dependence upon the data of the classical solution are shown by
using the generalized Fourier method. This paper also investigates the
inverse problem of finding a time-dependent coefficient of the heat equation
from the data of integral overdetermination condition.
\end{abstract}

\maketitle

\section{Introduction}

Let $T>0$ be a fixed number and $D_{T}=\{(x,t):0<x<1,$ $0<t\leq T\}$.

Consider the following initial-boundary value problem for the heat equation
in $\bar{D}_{T}$:

\begin{equation}
u_{t}=u_{xx}-p(t)u+f(x,t),\text{ }
\end{equation}%
with the initial condition

\begin{equation}
u(x,0)=\varphi (x),\text{ }
\end{equation}%
and the boundary conditions

\begin{equation}
u(0,t)=0,\text{ \ }au_{xx}(1,t)+du_{x}(1,t)-bu(1,t)=0,
\end{equation}%
where $f,$ $\varphi $ are given functions when $0\leq t\leq T$, $0\leq x\leq
1$ and $a,b,d$ are given numbers.

When the coefficient $p(t),$ $0\leq t\leq T$ is also given, the problem of
finding $u(x,t)$ from using equation (1.1), initial condition (1.2) and
boundary conditions (1.3) is termed as the direct (or forward) problem.

This problem can be used in a heat transfer and diffusion processes where a
source parameter is present. Taking into account the equation at $x=1$, in
the case $a\neq 0$, the second boundary condition becomes to the form of
dynamically boundary condition as 
\begin{equation*}
au_{t}(1,t)+du_{x}(1,t)+(ap(t)-b)u(1,t)=af(1,t).
\end{equation*}%
This boundary condition is observed in the process of cooling of a thin
solid bar one end of which is placed in contact in the case of perfect
thermal contact [1, p. 262]. Another possible application of such boundary
condition is announced in [2, p. 79] which represents a boundary reaction in
diffusion of chemical, where the term $du_{x}(1,t)$ represents the diffusive
transport of materials to the boundary.

When the function $p(t),$ $0\leq t\leq T$ is unknown, the inverse problem
formulates a problem of finding a pair of functions $\left\{ p(t),\text{ }%
u(x,t)\right\} $ such that satisfy the equation (1.1), initial condition
(1.2), boundary conditions (1.3) and overdetermination condition%
\begin{equation}
\int\limits_{0}^{1}u(x,t)dx=E(t),0\leq t\leq T.
\end{equation}

If we let $u(x,t)$ present the temperature distribution, then
above-mentioned inverse problem can be regarded as a control problem with
source control. The source control parameter $p(t)$ needs to be determined
from thermal energy $E(t)$.

Because the function $p$ is space independent, $a,b$ and $d$ are constants
and the boundary conditions are linear and homogeneous, the method of
separation of variables is suitable for studying the problems under
consideration. It is well known that, the main difficulty for applying
Fourier method is its basisness, i.e. expansion in terms of eigenfunctions
of auxiliary spectral problem 
\begin{equation}
\left\{ 
\begin{array}{l}
y^{\prime \prime }(x)+\lambda y(x)=0,\text{ }0\leq x\leq 1, \\ 
\text{ }y(0)=0, \\ 
\left( a\lambda +b\right) y(1)=dy^{^{\prime }}(1).%
\end{array}%
\right.
\end{equation}

In contrast to classical Sturm-Liouville problem, this problem has spectral
parameter also in boundary condition. It makes impossible to apply the
classical results to the expansion in terms of eigenfunctions [3,4]. The
spectral analysis of such type of problems started by Walter [5]. The
important developments are made by Fulton [6], Kerimov, Allakhverdiev [7,8],
Binding, Browne, Seddighi [9], Kapustin, Moiseev [10], Kerimov, Poladov
[11]. It is useful to note the paper [12] for which the results on expansion
in terms of eigenfunctions are firmly used in present paper.

The inverse problem of finding the coefficient $p(t)$ in the equation (1.1)
with the nonlocal boundary conditions are considered in the papers [13-16].
In contrast to these papers, in present paper the boundary conditions are
localized to the points $x=0$ and $x=1$. The literature devoted to inverse
problems of a finding time-dependent coefficient for the equation (1.1) with
the localized boundary conditions are so vast, see [17-19], to name only a
few references. The principal difference the boundary conditions in present
paper from the others localized boundary conditions is that the existence
the term $u_{xx}(1,t)$. As is noted, this boundary condition is reduced to a
dinamical type boundary condition by using the expression of the equation
(1.1).

The paper is organized as follows. In Chapter 2, the eigenvalues and
eigenfunctions of the auxiliary spectral problem and some of their
properties are introduced. In Chapter 3, the existence, uniqueness and the
continuous dependence upon the data of the solution of direct problem
(1.1)-(1.3) is proved. Finally in Chapter 4, the existence, uniqueness and
continuous dependence upon the data of the solution of the inverse problem
(1.1)-(1.4) is shown.

\section{Some properties of the auxiliary spectral problem}

Consider the spectral problem (1.4) with $ad>0$. It is known in [9] that,
eigenvalues $\lambda _{n},$ $n=0,1,2,...$ are real and simple. They are
unbounded increasing sequence 
\begin{equation*}
\lambda _{0}<\lambda _{1}<\cdots <\lambda _{n}<\cdots ,
\end{equation*}%
and the eigenfunction $y_{n}(x)$ corresponding to $\lambda _{n}$ has exactly 
$n$ simple zeros in the interval $(0,1)$.

For the positive eigenvalues $\lambda _{n}=\mu _{n}^{2}$, $\mu _{n}$ are the
simple positive roots of characteristic equation $(\frac{a}{d}\mu ^{2}+\frac{%
b}{d})\sin \mu =\mu \cos \mu $. The placement of all eigenvalues with
respect to $\lambda =0$ are as follows:%
\begin{eqnarray*}
\lambda _{0} &<&0<\lambda _{1}<\lambda _{2}<\cdots ,\text{ for }\frac{b}{d}%
>1, \\
\lambda _{0} &=&0<\lambda _{1}<\lambda _{2}<\cdots ,\text{ for }\frac{b}{d}%
=1, \\
0 &<&\lambda _{0}<\lambda _{1}<\lambda _{2}<\cdots ,\text{ for }\frac{b}{d}<1
\end{eqnarray*}

In the case $\frac{b}{d}>1$, the eigenfunctions are $y_{0}(x)=e^{\mu
_{0}x}-e^{-\mu _{0}x},$ $y_{n}(x)=\sin (\mu _{n}x),$ $n=1,2,...,$ which
correspond to eigenvalues $\lambda _{0}=-\mu _{0}^{2}$, $\lambda _{n}=\mu
_{n}^{2},n=1,2,...$. Besides, $\pi n<$ $\mu _{n}<\frac{\pi }{2}+\pi n$.

In the case $\frac{b}{d}=1$, the eigenfunctions are $y_{0}(x)=x,$ $%
y_{n}(x)=\sin (\mu _{n}x),$ $n=1,2,...,$ which correspond to eigenvalues $%
\lambda _{0}=0$, $\lambda _{n}=\mu _{n}^{2},n=1,2,...$. Besides, $\pi n<$ $%
\mu _{n}<\frac{\pi }{2}+\pi n$.

In the case $\frac{b}{d}<1$, the eigenfunctions are $y_{n}(x)=\sin (\mu
_{n}x),$ $n=0,1,2,...,$ which correspond to eigenvalues $\lambda _{n}=\mu
_{n}^{2},$ $n=0,1,2,...$. In addition, $\pi n<$ $\mu _{n}<\frac{\pi }{2}+\pi
n$ for $0\leq \frac{b}{d}<1$ and $\pi n<$ $\mu _{n}<\pi +\pi n$ for $\frac{b%
}{d}<0$.

It is easy to see that, in all of these cases 
\begin{equation}
\int_{0}^{1}y_{n}(x)dx>0,\text{ }n=0,1,2,...\text{.}
\end{equation}

The detailed investigation of the characteristic equation allows us to
determine the following asymptotic of eigenfunctions and eigenvalues:

\begin{eqnarray}
\mu _{n} &=&\pi n+\frac{d}{a\pi n}+O(\frac{1}{n^{3}}),\text{ }  \notag \\
y_{n}(x) &=&\sin (\pi nx)+\left[ \frac{d}{a\pi n}\cos (\pi nx)+O(\frac{1}{%
n^{3}})\right] x,
\end{eqnarray}%
for a sufficiently large $n$.

Let $n_{0}$ be arbitrary fixed nonnegative integer. It is shown in [12] that
the system of eigenfunctions $\left\{ y_{n}(x)\right\} $ $(n=0,1,2,...;n\neq
n_{0})$ is a Riesz basis for $\mathbf{L}_{2}\left[ 0,1\right] $. The system $%
\left\{ u_{n}(x)\right\} $ $(n=0,1,2,...;n\neq n_{0}),$ biortogonal to the
system $\left\{ y_{n}(x)\right\} $ $(n=0,1,2,...;n\neq n_{0})$ has the form 
\begin{equation*}
u_{n}(x)=\frac{y_{n}(x)-\frac{y_{n}(1)}{y_{n_{0}}(1)}y_{n_{0}}(x)}{%
\left\Vert y_{n}\right\Vert _{\mathbf{L}_{2}\mathbf{[}0,1\mathbf{]}}^{2}+%
\frac{a}{d}y_{n}^{2}(1)}\text{,}
\end{equation*}%
that is $\left( y_{n},u_{m}\right) =$ $\delta _{n,m},$ $n,m=0,1,2,...;n,m%
\neq n_{0}$ where $\delta _{n,m}$ is the Kronecker symbol.

The following lemma is true for the systems $\left\{ y_{n}(x)\right\} ,$ $%
\left\{ u_{n}(x)\right\} $ $(n=0,1,2,...;n\neq n_{0})$.

\begin{lemma}
Let $\varphi (x)\in \mathbf{C}^{3}\left[ 0,1\right] $ be arbitrary function
satisfying the conditions 
\begin{equation}
\varphi (0)=\varphi ^{\prime \prime }(0)=0,\text{ }\varphi (1)=\varphi
^{\prime }(1)=\varphi ^{\prime \prime }(1)=0
\end{equation}%
and 
\begin{equation}
\int_{0}^{1}\varphi (x)y_{n_{0}}(x)dx=0.
\end{equation}%
Then the inequalities 
\begin{eqnarray}
\dsum\limits_{\underset{(n\neq n_{0})}{n=0}}^{\infty }\left\vert \lambda
_{n}(\varphi ,y_{n})\right\vert &\leq &C_{0}\left\Vert \varphi \right\Vert _{%
\mathbf{C}^{3}\mathbf{[}0,1\mathbf{]}},\text{ }\dsum\limits_{\underset{%
(n\neq n_{0})}{n=0}}^{\infty }\left\vert \lambda _{n}(\varphi
,u_{n})\right\vert \leq C\left\Vert \varphi \right\Vert _{\mathbf{C}%
^{3}[0,1]} \\
&&\text{ (}C_{0}\text{ and }C\text{ are constants)}  \notag
\end{eqnarray}%
hold.
\end{lemma}

\begin{proof}
The asymptotic representation (2.2) implies 
\begin{equation*}
\left\Vert y_{n}\right\Vert _{\mathbf{L}_{2}\mathbf{[}0,1\mathbf{]}}^{2}=%
\frac{1}{2}+O(\frac{1}{n^{2}}),\text{ }y_{n}^{2}(1)=O(\frac{1}{n^{2}}).
\end{equation*}%
Then there exist constants $\alpha $ and $\beta $, such that 
\begin{equation*}
\beta >\left\Vert y_{n}\right\Vert _{\mathbf{L}_{2}\mathbf{[}0,1\mathbf{]}%
}^{2}+\frac{a}{d}y_{n}^{2}(1)>\alpha >0.
\end{equation*}%
According to last inequality, under condition (2.4), the convergence of
series $\dsum\limits_{\underset{(n\neq n_{0})}{n=0}}^{\infty }\left\vert
\lambda _{n}(\varphi ,y_{n})\right\vert $ is equivalent to convergence of $\
\dsum\limits_{\underset{(n\neq n_{0})}{n=0}}^{\infty }$ $\left\vert \lambda
_{n}(\varphi ,u_{n})\right\vert =\dsum\limits_{\underset{(n\neq n_{0})}{n=0}%
}^{\infty }\frac{\left\vert \lambda _{n}(\varphi ,y_{n})\right\vert }{%
\left\Vert y_{n}\right\Vert _{\mathbf{L}_{2}\mathbf{[}0,1\mathbf{]}}^{2}+%
\frac{a}{d}y_{n}^{2}(1)}$. In addition, this inequality allows us to obtain
second estimate in (2.5) if we know first one.

Since $\lambda _{n}y_{n}=-y_{n}^{\prime \prime }$ and $y_{n}(0)=0$, the
equality $\lambda _{n}(\varphi ,y_{n})=-(\varphi ^{\prime \prime },y_{n})$
is obtained with helping two time integration by parts and using (2.3).
Taking into account the representation (2.2) we obtain 
\begin{equation}
\lambda _{n}(\varphi ,y_{n})=-(\varphi ^{\prime \prime },\sin (\pi nx))-%
\frac{d}{a\pi n}(x\varphi ^{\prime \prime },\cos (\pi nx))+\left\Vert
\varphi ^{\prime \prime }\right\Vert _{\mathbf{C}\left[ 0,1\right] }O(\frac{1%
}{n^{2}}).
\end{equation}%
By using Schwarz and Bessel inequalities, it is easy to show that 
\begin{eqnarray}
\dsum\limits_{\underset{(n\neq n_{0})}{n=0}}^{\infty }\frac{d}{a\pi n}%
\left\vert (x\varphi ^{\prime \prime },\cos (\pi nx))\right\vert &\leq
&const\left\Vert x\varphi ^{\prime \prime }\right\Vert _{\mathbf{L}_{2}%
\mathbf{[}0,1\mathbf{]}}\leq const\left\Vert \varphi ^{\prime \prime
}\right\Vert _{\mathbf{C[}0,1\mathbf{]}},  \notag \\
&& \\
\dsum\limits_{\underset{(n\neq n_{0})}{n=1}}^{\infty }\left\Vert \varphi
^{\prime \prime }\right\Vert _{\mathbf{C}\left[ 0,1\right] }O(\frac{1}{n^{2}}%
) &\leq &const\left\Vert \varphi ^{\prime \prime }\right\Vert _{\mathbf{C[}%
0,1\mathbf{]}}.  \notag
\end{eqnarray}%
In addition, by using integration by parts and Schwarz and Bessel
inequalities we obtain that 
\begin{equation}
\dsum\limits_{\underset{(n\neq n_{0})}{n=0}}^{\infty }\left\vert (\varphi
^{\prime \prime },\sin (\pi nx))\right\vert \leq const\left\Vert \varphi
^{\prime \prime \prime }\right\Vert _{\mathbf{C[}0,1\mathbf{]}}.
\end{equation}%
Putting (2.6) and (2.7) in (2.8) yields the inequality (2.5).
\end{proof}

\bigskip Let us introduce the following notation for the simplicity.

\begin{notation}
The class of functions which satisfy the conditions of the Lemma 1 we will
denote as $\mathbf{\Phi }_{n_{0}}$, that is,%
\begin{equation*}
\mathbf{\Phi }_{n_{0}}\equiv \left\{ 
\begin{array}{c}
\varphi (x)\in \mathbf{C}^{3}\left[ 0,1\right] :\varphi (0)=\varphi ^{\prime
\prime }(0)=0,\text{ }\varphi (1)=\varphi ^{\prime }(1)=\varphi ^{\prime
\prime }(1)=0,\text{ } \\ 
\int_{0}^{1}\varphi (x)y_{n_{0}}(x)dx=0%
\end{array}%
\right\}
\end{equation*}
\end{notation}

Because the series $\dsum\limits_{\underset{(n\neq n_{0})}{n=0}}^{\infty
}\left\vert \lambda _{n}(\varphi ,y_{n})\right\vert $ is majorant for the
series $\dsum\limits_{\underset{(n\neq n_{0})}{n=0}}^{\infty }\left\vert
(\varphi ,y_{n})\right\vert $, the following corollary of Lemma 1 hold.

\begin{corollary}
For arbitrary $\varphi (x)\in \Phi _{n_{0}}$ the estimates 
\begin{equation*}
\dsum\limits_{\underset{(n\neq n_{0})}{n=0}}^{\infty }\left\vert (\varphi
,y_{n})\right\vert \leq \bar{C}_{0}\left\Vert \varphi \right\Vert _{\mathbf{C%
}^{3}\mathbf{[}0,1\mathbf{]}},\text{ }\dsum\limits_{\underset{(n\neq n_{0})}{%
n=0}}^{\infty }\left\vert (\varphi ,u_{n})\right\vert \leq \bar{C}\left\Vert
\varphi \right\Vert _{\mathbf{C}^{3}\mathbf{[}0,1\mathbf{]}}
\end{equation*}%
hold, where $\bar{C}_{0}$ and $\bar{C}$ are constants.
\end{corollary}

\section{Classical solution of the direct problem}

Let $p(t)\in \mathbf{C[}0,T\mathbf{]}$ be a known continuous function. The
function $u(x,t)\ $from the class\ $\mathbf{C}^{2,0}\left( \bar{D}%
_{T}\right) \mathbf{\cap C}^{2,1}\left( D_{T}\right) $\ that satisfy (1.1)
in $D_{T}$, the initial condition (1.2) and the boundary condition (1.3) is
said to be classical solution of the mixed problem (1.1)-(1.3).

The smootness conditions $f\in \mathbf{C}\left( D_{T}\right) ,$ $\varphi \in 
\mathbf{C}^{2}\mathbf{[}0,1\mathbf{]}$ and the consistency conditions 
\begin{equation}
\varphi (0)=0,\text{ }a\varphi ^{\prime \prime }(1)+d\varphi ^{\prime
}(1)-b\varphi (1)=0
\end{equation}%
are the necessary conditions for the existence of a classical solution of
the problem (1.1)-(1.3).

To construct the formal solution of the problem (1.1)-(1.3) we will use
generalized Fourier method. In accordance with this method, the solution $%
u(x,t)$ is sought in a Fourier series in term of the eigenfunctions $\left\{
y_{n}(x)\right\} $ $(n=0,1,2,...;n\neq n_{0})$ of auxiliary spectral problem
(1.5):%
\begin{equation*}
u(x,t)=\dsum\limits_{\underset{(n\neq n_{0})}{n=0}}^{\infty
}v_{n}(t)y_{n}(x),\text{ }v_{n}(t)=(u,u_{n}).
\end{equation*}%
For the functions $v_{n}(t),n=0,1,2,...;n\neq n_{0}$ we obtain the Cauchy
problem%
\begin{eqnarray*}
v_{n}^{\prime }(t)+(p(t)+\lambda _{n})v_{n}(t) &=&f_{n}(t), \\
v_{n}(0) &=&\varphi _{n},\text{ }
\end{eqnarray*}%
where $f_{n}(t)=(f,y_{n}),\varphi _{n}=(\varphi ,u_{n})$.

Solving these Cauchy problems, we obtain 
\begin{equation*}
v_{n}(t)=\varphi _{n}e^{-\int_{0}^{t}[p(s)+\lambda
_{n}]ds}+\int_{0}^{t}f_{n}(s)e^{-\int_{s}^{t}[p(\tau )+\lambda _{n}]d\tau }ds
\end{equation*}%
and the formal solution of the mixed problem (1.1)-(1.3) is expressed via
the series%
\begin{equation}
u(x,t)=\dsum\limits_{\underset{(n\neq n_{0})}{n=1}}^{\infty }\left[ \varphi
_{n}e^{-\lambda _{n}t-\int_{0}^{t}p(s)ds}+\int_{0}^{t}f_{n}(s)e^{-\lambda
_{n}(t-s)-\int_{s}^{t}p(\tau )d\tau }ds\right] y_{n}(x).
\end{equation}

Now we prove a theorem on the existence of the classical solution of the
problem (1.1)-(1.3). \bigskip

\begin{theorem}
(Existence) If $p(t)\in \mathbf{C[}0,T\mathbf{]},$ $f(x,t)\in \mathbf{C}%
\left( \bar{D}_{T}\right) ,\varphi (x)\in \mathbf{\Phi }_{n_{0}}$ and $%
f(x,t)\in \mathbf{\Phi }_{n_{0}}$ for every $t\in \lbrack 0,T]$. Then the
series (3.2) gives a classical solution of the problem (1.1)-(1.3) and $%
u(x,t)$ belongs to $\mathbf{C}^{2,1}\left( \bar{D}_{T}\right) $.
\end{theorem}

\begin{proof}
By the hypotesis, we obtain from Lemma 1 that the series 
\begin{equation}
M_{1}\dsum\limits_{\underset{(n\neq n_{0})}{n=0}}^{\infty }\left\vert
\lambda _{n}\varphi _{n}\right\vert \text{ and }M_{2}\dsum\limits_{\underset{%
(n\neq n_{0})}{n=0}}^{\infty }\left\vert \lambda _{n}\right\vert
\int_{0}^{T}\left\vert f_{n}(s)\right\vert ds
\end{equation}%
($M_{1}$ and $M_{2}$ are some constants) are convergent. These series are
the majorizing series of (3.2) and its $x-$partial, $xx-$partial
derivatives. The majorizing series for the $x-$partial derivative of (3.2)
follows from the fact that the functions $y_{n}^{\prime }(x),n>2$ have at
least one zero in $(0,1)$, since the eigenfunction $y_{n}(x)$ corresponding
to $\lambda _{n}$ has exactly $n$ simple zeros in the interval $(0,1)$. In
this case, the equality $\lambda _{n}y_{n}=-y_{n}^{\prime \prime }$ implies $%
y_{n}^{\prime }(x)=$ $\lambda _{n}\int_{x}^{\omega }y_{n}(s)ds$, where $%
\omega \in (0,1)$ are a zero of $y_{n}^{\prime }(x)$. Because the sequence $%
y_{n}(x),n=1,2,\ldots $ is uniformly bounded, the series $\dsum\limits_{%
\underset{(n\neq n_{0})}{n=0}}^{\infty }\left\vert \varphi _{n}\right\vert
\left\vert y_{n}^{\prime }(x)\right\vert $ can be majorizing by $%
M_{3}\dsum\limits_{\underset{(n\neq n_{0})}{n=0}}^{\infty }\left\vert
\varphi _{n}\right\vert ,$ where $M_{3}$ is some constant. Thus, the series
(3.2) and its $x-$partial, $xx-$partial derivatives are uniformly convergent
in $\bar{D}_{T}$. Since these series uniformly convergent, their sums $%
u(x,t) $, $u_{x}(x,t)$ and $u_{xx}(x,t)$ are continuos in $\bar{D}_{T}$. The
series (3.3) are also majorizing series for the $t-$partial derivative of
the (3.2). Therefore $u_{t}(x,t)$ also continuous in $\bar{D}_{T}$. Thus, $%
u(x,t)\in \mathbf{C}^{2,1}\left( \bar{D}_{T}\right) $ and satisfies the
condition (1.1)-(1.3) by the superposition principle.

The proof of the theorem is complete.
\end{proof}

\begin{remark}
The existence of the classical solution of the problem (1.1)-(1.3) can be
obtained for each of the classes of functions $\mathbf{\Phi }_{n_{0}}$, $%
n_{0}=0,1,2,\ldots $, therefore, the classical solution exists for the union 
$\mathbf{\Phi }\equiv \dbigcup\limits_{n_{0}=0}^{\infty }\mathbf{\Phi }%
_{n_{0}}$. It can be expected that the integral condition (2.4) is not
essential for the existence of the classical solution.
\end{remark}

\begin{remark}
\bigskip The uniqueness of the solution of the problem (1.1)-(1.3), under
the conditions of the Theorem 1 is obtained from the uniqueness of the
representation (3.2). The uniqueness of the solution under smoothness and
consistency conditions (3.1) is obtained by using maximum-minimum principle.
It will be done in next theorem.
\end{remark}

\begin{lemma}
\bigskip (maximum-minimum principle) Let $p(t)\geq 0,t\in \lbrack 0,T]$, and 
$u(x,t)\in $\ $\mathbf{C}^{2,0}\left( \bar{D}_{T}\right) \cap \mathbf{C}%
^{2,1}\left( D_{T}\right) $ satisfies the equation (1.1) in $D_{T}$. If $%
f(x,t)\leq 0$ in $D_{T}$ then 
\begin{equation*}
u(x,t)\leq \max \left\{ 0,\underset{0\leq x\leq 1}{\max }u(x,0),\underset{%
0\leq t\leq T}{\max }u(0,t),\underset{0\leq t\leq T}{\max }u(1,t)\right\} .
\end{equation*}%
If $f(x,t)\geq 0$ in $D_{T}$ then 
\begin{equation*}
u(x,t)\geq \min \left\{ 0,\underset{0\leq x\leq 1}{\min }u(x,0),\underset{%
0\leq t\leq T}{\min }u(0,t),\underset{0\leq t\leq T}{\min }u(1,t)\right\} .
\end{equation*}
\end{lemma}

The proof of this lemma is omitted because it can be found in [20, p. 390].

\begin{theorem}
(Uniqueness and continuous dependence upon the data)The classical solution
of the problem (1.1)-(1.3), with $p(t)\geq 0$ is unique and depends
continuously on $f(x,t)\in \mathbf{C}\left( \bar{D}_{T}\right) $ and $%
\varphi (x)\in \mathbf{C}^{2}[0,1]$ in the sense that 
\begin{equation}
\left\Vert u-\tilde{u}\right\Vert _{\mathbf{C}\left( \bar{D}_{T}\right)
}\leq \left\Vert \varphi -\tilde{\varphi}\right\Vert _{\mathbf{C}%
[0,1]}+(1+\left\vert b\right\vert )T\left\Vert f-\tilde{f}\right\Vert _{%
\mathbf{C}\left( \bar{D}_{T}\right) },
\end{equation}%
where $u(x,t)$ and $\tilde{u}(x,t)$ are the classical solutions of
(1.1)-(1.3) with the data $f,\varphi $ and $\tilde{f},\tilde{\varphi}$,
respectively.

\begin{proof}
Let $u(x,t)$ be the classical solution of the problem (1.1)-(1.3). Introduce
the notations 
\begin{equation*}
R=\left\Vert f\right\Vert _{\mathbf{C}\left( \bar{D}_{T}\right) },\text{ }%
K=\left\Vert \varphi \right\Vert _{\mathbf{C}[0,1]}\text{.}
\end{equation*}%
Construct the function 
\begin{equation*}
g(x,t)=u(x,t)-Rt\text{.}
\end{equation*}%
This function is a classical solution of the mixed problem 
\begin{equation*}
u_{t}=u_{xx}-p(t)u+f(x,t)-R-p(t)Rt,\text{ }
\end{equation*}%
\begin{equation*}
u(x,0)=\varphi (x),\text{ }
\end{equation*}%
\begin{equation*}
u(0,t)=-Rt,\text{ }au_{xx}(1,t)+du_{x}(1,t)-bu(1,t)=bRt.
\end{equation*}

Allowing for $f(x,t)-R-p(t)Rt\leq 0,$ $bRt\leq \left\vert b\right\vert RT$
and using maximum principle, we obtain the estimate $g(x,t)\leq \max \left\{
K,\left\vert b\right\vert RT\right\} $, i.e. 
\begin{equation*}
u(x,t)\leq \max \left\{ K,\left\vert b\right\vert RT\right\} +RT\leq
K+(1+\left\vert b\right\vert )RT\text{.}
\end{equation*}%
Similarly, if we introduce the function 
\begin{equation*}
h(x,t)=u(x,t)+Rt
\end{equation*}%
and use the minimum principle, we arrive at a opposite estimate:%
\begin{equation*}
u(x,t)\geq -\max \left\{ K,\left\vert b\right\vert RT\right\} -RT\geq
-K-(1+\left\vert b\right\vert )RT\text{.}
\end{equation*}%
Thus, if $u(x,t)$ is the classical solution of the problem (1.1)-(1.3) we
have the estimate%
\begin{equation}
\left\Vert u\right\Vert _{\mathbf{C}\left( \bar{D}_{T}\right) }\leq
\left\Vert \varphi \right\Vert _{\mathbf{C}[0,1]}+(1+\left\vert b\right\vert
)T\left\Vert f\right\Vert _{\mathbf{C}\left( \bar{D}_{T}\right) }.
\end{equation}%
The uniqueness follows from the fact that by virtue of the estimate (3.5),
the homogeneous problem (1.1)-(1.3) (i.e. with $f=0,$ and $\varphi =0$) has
only a zero classical solution.

To prove continuous dependence on data we study the difference $%
v(x,t)=u(x,t)-\tilde{u}(x,t)$. This function is a classical solution of
(1.1)-(1.3) with $f-\tilde{f}$, and $\varphi -\tilde{\varphi}$ substituted
for $f$ and $\varphi $, respectively. Applying the inequality (3.5) to $%
v(x,t)$ we arrive at the estimate (3.4).
\end{proof}
\end{theorem}

\section{Classical solution of the inverse problem}

Let $p(t),t\in \lbrack 0,T]$ be a unknown function. The pair $%
\{p(t),u(x,t)\} $ from the class $\mathbf{C}[0,T]\times (\mathbf{C}^{2,1}(%
\bar{D}_{T})\cap \mathbf{C}^{2,0}(D_{T}))$ for which the conditions
(1.1)--(1.4) are satisfied, is called a classical solution of the inverse
problem (1.1)--(1.4).

We have the following assumptions on $\varphi ,$ $E$ and $f.$\bigskip

\begin{enumerate}
\item[$($A$_{1})$] $%
\begin{array}{cccc}
(\text{A}_{1})_{1} & \varphi (x)\in \mathbf{\Phi }_{n_{0}}; & (\text{A}%
_{1})_{2} & \varphi _{0}>0,\varphi _{n}\geq 0,n=0,1,2,...(n\neq n_{0}),%
\end{array}%
$

\item[$($A$_{2})$] $%
\begin{array}{cc}
(\text{A}_{2})_{1} & E(t)\in \mathbf{C}^{1}\left[ 0,T\right] ;\text{ }%
E(0)=\int\limits_{0}^{1}\varphi (x)dx; \\ 
(\text{A}_{2})_{2} & \text{ }E(t)>0,\forall t\in \left[ 0,T\right] ;%
\end{array}%
$

\item[$($A$_{3})$] \bigskip $%
\begin{array}{cc}
(\text{A}_{3})_{1} & f(x,t)\in \mathbf{C}\left( \overline{D}_{T}\right) ;%
\text{ }f(x,t)\in \mathbf{\Phi }_{n_{0}},\text{ }\forall t\in \left[ 0,T%
\right] ; \\ 
(\text{A}_{3})_{2} & f_{n}(\tau )\geq 0,\text{ }n=0,1,2,...;n\neq n_{0};%
\end{array}%
$

where $\varphi _{n}=\int\limits_{0}^{1}\varphi (x)u_{n}(x)dx,$ $%
f_{n}(t)=\int\limits_{0}^{1}f(x,t)u_{n}(x)dx,$ $n=0,1,2,...$.
\end{enumerate}

The main result is presented as follows.

\begin{theorem}
(Existence and uniqueness) Let $\left( \text{A}_{1}\right) -\left( \text{A}%
_{3}\right) $ be satisfied. Then the inverse problem (1.1)-(1.4) has a
unique classical solution.
\end{theorem}

\begin{proof}
We already know that the solution of the mixed problem (1.1)-(1.3) is
expressed via the series%
\begin{equation}
u(x,t)=\dsum\limits_{\underset{(n\neq n_{0})}{n=0}}^{\infty }\left[ \varphi
_{n}e^{-\lambda _{n}t-\int_{0}^{t}p(s)ds}\right] y_{n}(x)+\dsum\limits_{%
\underset{(n\neq n_{0})}{n=0}}^{\infty }\left[ \int_{0}^{t}f_{n}(s)e^{-%
\lambda _{n}(t-s)-\int_{s}^{t}p(\tau )d\tau }ds\right] y_{n}(x).
\end{equation}%
for arbitrary $p(t)\in \mathbf{C}\left[ 0,T\right] $. In addition $u(x,t)\in 
\mathbf{C}^{2,1}\left( \bar{D}_{T}\right) $.

Applying the overdetermination condition (1.4), we obtain the following
Volterra integral equation of the second kind with respect to $%
q(t)=e^{\int\limits_{0}^{t}p(s)ds}$:%
\begin{equation}
q(t)=F(t)+\dint\limits_{0}^{t}K(t,\tau )q(\tau )d\tau ,
\end{equation}%
where 
\begin{eqnarray}
F(t) &=&\frac{1}{E(t)}\dsum\limits_{\underset{(n\neq n_{0})}{n=0}}^{\infty }%
\left[ \varphi _{n}e^{-\lambda _{n}t}\dint\limits_{0}^{1}y_{n}(x)dx\right] ,
\notag \\
&& \\
K(t,\tau ) &=&\frac{1}{E(t)}\dsum\limits_{\underset{(n\neq n_{0})}{n=0}%
}^{\infty }\left[ f_{2n-1}(\tau )e^{-\lambda _{n}\left( t-\tau \right)
}\dint\limits_{0}^{1}y_{n}(x)dx\right] .  \notag
\end{eqnarray}

In the case of existence of the positive solution of (4.2) in class $\mathbf{%
C}^{1}\left[ 0,T\right] $, the function $p(t)$ can be determined from $%
q(t)=e^{\int\limits_{0}^{t}p(s)ds}$ as 
\begin{equation}
p(t)=\frac{q^{\prime }(t)}{q(t)}.
\end{equation}

By using the boundness of the sequence $\dint%
\limits_{0}^{1}y_{n}(x)dx,n=0,1,2,\ldots $ and inequality (2.5) in Lemma 1,
under the assumptions $({A}_{1})_{1}-$ $({A}_{3})_{1},$ the right-hand side $%
F(t)$ and the kernel $K(t,\tau )$ are continuously differentiable functions
in $[0,T]$ and $[0,T]\times \lbrack 0,T],$ respectively. In addition,
according to the assumptions $({A}_{1})_{2}-({A}_{3})_{2}$ and formula
(2.1), the conditions $F(t)>0$ and $K(t,\tau )\geq 0$ are satisfied in $%
[0,T] $ and $[0,T]\times \lbrack 0,T],$ respectively.

In addition, the solution of (4.2) is given by the series 
\begin{equation*}
q(t)=\sum\limits_{n=0}^{\infty }(\mathbf{K}^{n}F)(t),
\end{equation*}%
where $(\mathbf{K}F)(t)\equiv \dint\limits_{0}^{t}K(t,\tau )F(\tau )d\tau $.
It is easy to verify that 
\begin{equation*}
\left\vert (\mathbf{K}^{n}F)(t)\right\vert \leq \left\Vert F\right\Vert _{%
\mathbf{C}\left[ 0,T\right] }\frac{\left( t\left\Vert K\right\Vert _{\mathbf{%
C}\left( [0,T]\times \lbrack 0,T]\right) }\right) ^{n}}{n!},\text{ \ }t\in %
\left[ 0,T\right] ,\text{ }n=0,1,\ldots \text{.}
\end{equation*}%
Thus, we obtain the estimate 
\begin{equation}
\left\Vert q\right\Vert _{\mathbf{C}\left[ 0,T\right] }\leq \left\Vert
F\right\Vert _{\mathbf{C}\left[ 0,T\right] }e^{T\left\Vert K\right\Vert _{%
\mathbf{C}\left( [0,T]\times \lbrack 0,T]\right) }}.
\end{equation}

We therefore obtain a unique positive function $q(t)$, continuously
differentiable in $[0,T]$, which the function (4.4) together with the
solution of the problem (1.1)-(1.3) given by the Fourier series (4.1), form
the unique solution of the inverse problem (1.1)-(1.4). Theorem 3 has been
proved.
\end{proof}

The following result on continuously dependence on the data of the solution
of the inverse problem (1.1)-(1.4) holds.

\begin{theorem}
(Continuous dependence upon the data) Let $\mathbf{\Im }$ be the class of
triples in the form of $T=\{f,$ $\varphi ,$ $E\}$ which satisfy the
assumptions $({A}_{1})-({A}_{3})$ of Theorem 3 and%
\begin{equation*}
\Vert f\Vert _{\mathbf{C}^{3,0}(\overline{D}_{T})}\leq N_{0},\Vert \varphi
\Vert _{\mathbf{C}^{3}[0,1]}\leq N_{1},\Vert E\Vert _{\mathbf{C}%
^{1}[0,T]}\leq N_{2},0<N_{3}\leq \underset{t\in \lbrack 0,T]}{\min }%
\left\vert E(t)\right\vert ,
\end{equation*}%
for some positive constants $N_{i},$ $i=0,$ $1,$ $2,$ $3$.

Then the solution pair $\left( u,p\right) $ of the inverse problem
(1.1)-(1.4) depends continuously upon the data in $\mathbf{\Im }$.
\end{theorem}

\begin{proof}
Let us denote $\left\Vert T\right\Vert =(\left\Vert E\right\Vert _{\mathbf{C}%
^{1}\left[ 0,T\right] }+\left\Vert \varphi \right\Vert _{\mathbf{C}^{3}\left[
0,1\right] }+\left\Vert f\right\Vert _{\mathbf{C}^{3,0}(\overline{D}_{T})})$.

Let $\ T=\{f,$ $\varphi ,$ $E\}$, $\ \bar{T}=\{\bar{f},$ $\bar{\varphi},$ $%
\bar{E}\}$ $\in $ $\Im $ be two sets of data.\ Let $\left( p,u\right) $ and $%
\left( \bar{p},\bar{u}\right) $ be solutions of inverse problems (1.1)-(1.4)
corresponding to the data $\Phi $ and $\bar{\Phi}$, respectively. Denote by $%
q(t)=e^{\int\limits_{0}^{t}p(s)ds}$, $\bar{q}(t)=e^{\int\limits_{0}^{t}\bar{p%
}(s)ds}$.

According to (4.1)and (4.2) we get:%
\begin{eqnarray}
q(t) &=&F(t)+\dint\limits_{0}^{t}K(t,\tau )q(\tau )d\tau ,  \notag \\
&& \\
F(t) &=&\frac{1}{E(t)}\dsum\limits_{\underset{(n\neq n_{0})}{n=0}}^{\infty }%
\left[ \varphi _{n}e^{-\lambda _{n}t}\dint\limits_{0}^{1}y_{n}(x)dx\right] ,
\notag \\
&&K(t,\tau )=\frac{1}{E(t)}\dsum\limits_{\underset{(n\neq n_{0})}{n=0}%
}^{\infty }\left[ f_{n}(\tau )e^{-\lambda _{n}\left( t-\tau \right)
}\dint\limits_{0}^{1}y_{n}(x)dx\right] ,  \notag
\end{eqnarray}

and

\begin{eqnarray}
\bar{q}(t) &=&\bar{F}(t)+\dint\limits_{0}^{t}\bar{K}(t,\tau )\bar{q}(\tau
)d\tau ,  \notag \\
&& \\
\bar{F}(t) &=&\frac{1}{\bar{E}(t)}\dsum\limits_{\underset{(n\neq n_{0})}{n=0}%
}^{\infty }\left[ \bar{\varphi}_{n}e^{-\lambda
_{n}t}\dint\limits_{0}^{1}y_{n}(x)dx\right] ,  \notag \\
&&\bar{K}(t,\tau )=\frac{1}{\bar{E}(t)}\left( \dsum\limits_{\underset{(n\neq
n_{0})}{n=0}}^{\infty }\left[ \bar{f}_{n}(\tau )e^{-\lambda _{n}\left(
t-\tau \right) }\dint\limits_{0}^{1}y_{n}(x)dx\right] \right) ,  \notag
\end{eqnarray}%
Taking into account the inequality in Corollary 1 the next inequalities will
be true: 
\begin{eqnarray}
\text{ }\left\vert F(t)\right\vert &\leq &\frac{M}{\left\vert
E(t)\right\vert }\left( \dsum\limits_{\underset{(n\neq n_{0})}{n=0}}^{\infty
}\left\vert \left( \varphi ,u_{n}\right) \right\vert \right) \leq \\
&\leq &\frac{M}{N_{3}}\left\Vert \varphi \right\Vert _{\mathbf{C}%
^{3}[0,1]}\leq \frac{M}{N_{3}}CN_{1},\text{ \ \ }  \notag \\
&&  \notag \\
\left\vert \bar{K}(t,\tau )\right\vert &\leq &\frac{M}{N_{3}}\underset{t\in
\lbrack 0,T]}{\max }\left\Vert \bar{f}(\cdot ,t)\right\Vert _{\mathbf{C}%
^{3}[0,1]}\leq \frac{M}{N_{3}}CN_{0},\text{ \ }  \notag
\end{eqnarray}%
where $C$ is constant mentioned in Lemma 1, $M$ is the constant which $%
0<\dint\limits_{0}^{1}y_{n}(x)dx\leq M$.

First let us estimate the difference $q-\bar{q}$. From (4.6) and (4.7) we
obtain: 
\begin{equation}
q(t)-\bar{q}(t)=F(t)-\bar{F}(t)+\dint\limits_{0}^{t}\left[ K(t,\tau )-\bar{K}%
(t,\tau )\right] q(\tau )d\tau +\dint\limits_{0}^{t}\bar{K}(t,\tau )\left[
q(\tau )-\bar{q}(\tau )\right] d\tau
\end{equation}

Let $\alpha =\left\Vert F-\bar{F}\right\Vert _{\mathbf{C}\left[ 0,T\right]
}+T\left\Vert K-K\right\Vert _{\mathbf{C}(\left[ 0,T\right] \times \left[ 0,T%
\right] )}\left\Vert q\right\Vert _{\mathbf{C}\left[ 0,T\right] }$.

Then denoting $Q(t)=\left\vert q(t)-\bar{q}(t)\right\vert $, equation (4.9)
implies the inequality%
\begin{equation}
Q(t)\leq \alpha +\dint\limits_{0}^{t}\left\vert \bar{K}(t,\tau )\right\vert
Q(t)dt\text{.}
\end{equation}%
Applying the Gronwall inequality ([21, p. 9]), from (4.10) we obtain that 
\begin{equation*}
Q(t)\leq \alpha e^{\dint\limits_{0}^{t}\underset{s\in \lbrack \tau ,t]}{\sup 
}\left\vert \bar{K}(s,\tau )\right\vert Q(t)dt}\text{.}
\end{equation*}%
Finally using (4.5) and (4.8), we obtain 
\begin{equation*}
\left\Vert q-\bar{q}\right\Vert _{\mathbf{C}\left[ 0,T\right] }\leq
C_{1}(\left\Vert F-\bar{F}\right\Vert _{\mathbf{C}\left[ 0,T\right]
}+C_{2}\left\Vert K-\bar{K}\right\Vert _{\mathbf{C}(\left[ 0,T\right] \times %
\left[ 0,T\right] )}),
\end{equation*}%
where $C_{1}=e^{\frac{M}{N_{3}}CN_{0}T},$ $C_{2}=\frac{MT}{N_{3}}CN_{1}C_{1}$%
. It will be seen from (4.8) that $q$ continuously dependents upon $F$ and $%
K $.

Let us show that $F$ and $K$ continuously dependents upon the data. It is
easy to compute, with helping Corollary 1, that

\begin{eqnarray*}
\left\vert \bar{F}(t)-F(t)\right\vert &=&\left\vert \dsum\limits_{\underset{%
(n\neq n_{0})}{n=0}}^{\infty }\left[ \frac{\bar{\varphi}_{n}}{\bar{E}(t)}-%
\frac{\varphi _{n}}{E(t)}\right] e^{-\lambda
_{n}t}\dint\limits_{0}^{1}y_{n}(x)dx\right\vert \leq M\dsum\limits_{\underset%
{(n\neq n_{0})}{n=0}}^{\infty }\left\vert (\frac{\bar{\varphi}}{\bar{E}(t)}-%
\frac{\varphi }{E(t)},u_{n})\right\vert \\
&\leq &\frac{M\bar{C}N_{1}}{N_{3}^{2}}\left\Vert E-\bar{E}\right\Vert _{_{%
\mathbf{C}\left[ 0,T\right] }}+\frac{M\bar{C}}{N_{3}}\Vert \varphi -\bar{%
\varphi}\Vert _{\mathbf{C}^{3}[0,1]}
\end{eqnarray*}

\begin{eqnarray*}
\left\vert \bar{K}(t,\tau )-K(t,\tau )\right\vert &\leq &M\dsum\limits_{%
\underset{(n\neq n_{0})}{n=0}}^{\infty }\left\vert \frac{\bar{f}_{n}(\tau )}{%
\bar{E}(t)}-\frac{f_{n}(\tau )}{E(t)}\right\vert \\
&\leq &\frac{M\bar{C}N_{0}}{N_{3}^{2}}\left\Vert E-\bar{E}\right\Vert _{_{%
\mathbf{C}\left[ 0,T\right] }}+\frac{M\bar{C}}{N_{3}}\left\Vert f-\bar{f}%
\right\Vert _{\mathbf{C}^{3,0}(\overline{D}_{T})},
\end{eqnarray*}

By using last inequalities we obtain: 
\begin{eqnarray*}
\left\Vert F-\bar{F}\right\Vert _{\mathbf{C}\left[ 0,T\right] } &\leq
&M_{1}(\left\Vert E-\bar{E}\right\Vert _{_{\mathbf{C}\left[ 0,T\right]
}}+\left\Vert \varphi -\bar{\varphi}\right\Vert _{_{\mathbf{C}%
^{3}[0,1]}}+\left\Vert f-\bar{f}\right\Vert _{_{\mathbf{C}^{3,0}(\overline{D}%
_{T})}})\leq M_{1}\left\Vert T-\bar{T}\right\Vert \\
\left\Vert K-\bar{K}\right\Vert _{\mathbf{C}(\left[ 0,T\right] \times \left[
0,T\right] )} &\leq &M_{2}\left\Vert T-\bar{T}\right\Vert
\end{eqnarray*}%
where $M_{1}$ and $M_{2}$ are constants that are determined by constants $%
N_{k},$ $k=0,...,3$, $M$ and $\bar{C}$. This means that $F$ and $K$
continuously dependent upon the data. Thus, $q$ also continuously dependents
upon the data, by (4.3).

Now, let us show that $q^{\prime }$ also depends continuously upon the data.
Differentiating (4.6) and (4.7) with respect to $t$, we can obtain following
representation:%
\begin{eqnarray*}
q^{\prime }(t) &=&F^{\prime }(t)+K(t,t)q(t)+\dint\limits_{0}^{t}K_{t}(t,\tau
)q(\tau )d\tau , \\
\bar{q}^{\prime }(t) &=&\bar{F}^{\prime }(t)+\bar{K}(t,t)\bar{q}%
(t)+\dint\limits_{0}^{t}\bar{K}_{t}(t,\tau )\bar{q}(\tau )d\tau .
\end{eqnarray*}

The following estimation holds:%
\begin{eqnarray}
\left\Vert q^{\prime }-\bar{q}^{\prime }\right\Vert _{\mathbf{C}\left[ 0,T%
\right] } &\leq &\left\Vert F^{\prime }-\bar{F}^{\prime }\right\Vert _{%
\mathbf{C}\left[ 0,T\right] }  \notag \\
&&+\left( \left\Vert K-\bar{K}\right\Vert _{\mathbf{C}\left( [0,T]\times
\lbrack 0,T]\right) }+T\left\Vert K_{t}-\bar{K}_{t}\right\Vert _{\mathbf{C}%
\left( [0,T]\times \lbrack 0,T]\right) }\right) \left\Vert q\right\Vert _{%
\mathbf{C}\left[ 0,T\right] } \\
&&+\left( \left\Vert \bar{K}\right\Vert _{\mathbf{C}\left( [0,T]\times
\lbrack 0,T]\right) }+T\left\Vert \bar{K}_{t}\right\Vert _{\mathbf{C}\left(
[0,T]\times \lbrack 0,T]\right) }\right) \left\Vert q-\bar{q}\right\Vert _{%
\mathbf{C}\left[ 0,T\right] },  \notag
\end{eqnarray}%
Taking into account the facts that $q$, $F$ and $K$ continuously dependent
upon the data, using the inequality (4.5) and inequality 
\begin{equation*}
\left\vert \bar{K}_{t}(t,\tau )\right\vert \leq \frac{N_{2}}{N_{3}^{2}}%
C\left\Vert \bar{f}\right\Vert _{\mathbf{C}^{3,0}(\overline{D}_{T})}+\frac{C%
}{N_{3}}\left\Vert \bar{f}\right\Vert _{\mathbf{C}^{3,0}(\overline{D}%
_{T})}\leq \left( \frac{N_{2}}{N_{3}^{2}}+\frac{1}{N_{3}}\right) CN_{0},%
\text{ }
\end{equation*}%
it will be seen that $q^{\prime }$ depends continuously upon the $F^{\prime
} $ and $K_{t}$. By using (2.5), we can obtain similar estimations for $%
\left\Vert F^{\prime }-\bar{F}^{\prime }\right\Vert _{\mathbf{C}\left[ 0,T%
\right] }$ and $\left\Vert K_{t}-\bar{K}_{t}\right\Vert _{\mathbf{C}\left(
[0,T]\times \lbrack 0,T]\right) }$, as 
\begin{eqnarray*}
\left\Vert F^{\prime }-\bar{F}^{\prime }\right\Vert _{\mathbf{C}\left[ 0,T%
\right] } &\leq &\frac{2MCN_{1}}{N_{3}^{2}}\left\Vert E-\bar{E}\right\Vert
_{_{\mathbf{C}^{1}\left[ 0,T\right] }}+\frac{MC}{N_{3}}\Vert \varphi -\bar{%
\varphi}\Vert _{\mathbf{C}^{3}[0,1]}, \\
\left\Vert K_{t}-\bar{K}_{t}\right\Vert _{\mathbf{C}\left( [0,T]\times
\lbrack 0,T]\right) } &\leq &\frac{2MCN_{0}}{N_{3}^{2}}\left\Vert E-\bar{E}%
\right\Vert _{_{\mathbf{C}^{1}\left[ 0,T\right] }}+\frac{MC}{N_{3}}%
\left\Vert f-\bar{f}\right\Vert _{\mathbf{C}^{3,0}(\overline{D}_{T})}.
\end{eqnarray*}%
Thus, using these inequalities from (4.11) we obtain 
\begin{equation*}
\left\Vert q^{\prime }-\bar{q}^{\prime }\right\Vert _{\mathbf{C}\left[ 0,T%
\right] }\leq M_{3}(\left\Vert E-\bar{E}\right\Vert _{_{\mathbf{C}^{1}\left[
0,T\right] }}+\Vert \varphi -\bar{\varphi}\Vert _{\mathbf{C}%
^{3}[0,1]}+\left\Vert f-\bar{f}\right\Vert _{\mathbf{C}^{3,0}(\overline{D}%
_{T})})\leq M_{3}\left\Vert T-\bar{T}\right\Vert ,
\end{equation*}%
where $M_{3}$ is constant that is determined by constants $N_{k},$ $%
k=0,...,3 $, $M$ and $C$. It means that $q^{\prime }$ depends continuously
upon the data as well.

The equality $p(t)=\frac{q^{\prime }(t)}{q(t)}$ implies continuously
dependence of $p$ upon the data. Using the similar what we demonstrated
above we can prove that $u$, which is given in (4.1), depends continuously
upon the data. Theorem 4 has been proved.
\end{proof}

\section{Conclusion}

In this paper, we consider the initial-boundary value problem for the heat
equation with a dinamic type boundary condition which is observed in the
process of cooling of a thin solid bar one end of which is placed in contact
in the case of perfect thermal contact and in a boundary reaction in
diffusion of chemical. The Fourier method on the eigenfunctions of an
auxiliary spectral problem with the boundary condition which is dependent on
spectral parameter, is suitable for studying the problem under
consideration. Under some regularity, consistency and orthogonality
conditions, the existence, uniqueness and continuously dependence upon the
data of the classical solution are shown. This paper also investigates the
inverse problem of finding a coefficient of the heat equation from integral
overdetermination condition's data.


\begin{thebibliography}{99}
\bibitem{[1]} R.E. Langer, \textit{A problem in diffusion or in the flow of
heat for a solid in contact with a fluid,} Tohoku Math. J., 35 (1932), pp.
360-375.

\bibitem{[2]} J.R. Cannon,\ \textit{The one-dimensional heat equation,}
Addison-Wesley Publishing Company, California, 1984.

\bibitem{[3]} E.C. Tichmarsh,\ \textit{Eigenfunction expansions associated
with second order differential equations I,} Oxford Univ. Press, 1962.

\bibitem{[4]} M.A. Naimark, \textit{Linear differential operators:
Elementary theory of linear differential operators,} Frederick Ungar
Publishing Co., New York, 1967.

\bibitem{[5]} J. Walter, \textit{Regular eigenvalue problems with eigenvalue
parameter in the boundary conditions,} Math. Z, 133 (1973), pp. 301-312.

\bibitem{[6]} C.T. Fulton, \textit{Two-point boundary value problems with
eigenvalue parameter contained in the boundary conditions,} Proc. Roy. Soc.
Edinburgh Sect. A., 77 (1977), pp. 293-388.

\bibitem{[7]} N.B. Kerimov, T.I. Allakhverdiev, \textit{On a certain
boundary value problem. I,} Differential Equations 29, no. 1 (1993), pp.
45--50.

\bibitem{[8]} N.B. Kerimov, T.I. Allakhverdiev, \textit{On a certain
boundary value problem. II,} Differential Equations 29, no. 6 (1993), pp.
814--821 .

\bibitem{[9]} P.A. Binding, P.J. Brown, and K. Seddeghi, \textit{%
Sturm-Lioville problems with eigenparameter dependent boundary conditions,}
Proc. Edinburgh Math. Soc., (2), 37 (1) (1993), pp. 57-72.

\bibitem{[10]} N.Yu. Kapustin, E.I. Moiseev,{\normalsize \ }\textit{Spectral
problems with the spectral parameter in the boundary condition,}
Differential Equations, 33, no. 1 (1997), pp. 115-119.

\bibitem{[11]} N.B. Kerimov, R. G. Poladov, \textit{Basis properties of the
system of eigenfunctions in the Sturm--Liouville problem with a spectral
parameter in the boundary conditions,} Doklady Mathematics, 85, 1 (2012),
pp. 8--13.

\bibitem{[12]} N.B. Kerimov, V.S. Mirzoev, \textit{On the basis properties
of one spectral problem with a spectral parameter in a boundary condition,}
Siberian Mathematical Journal, 44 (5) (2003), pp. 813-816.

\bibitem{[13]} J.R. Cannon, Y. Lin, S. Wang, Determination of source
parameter in parabolic equation, Meccanica 27 (1992) 85--94.

\bibitem{[14]} M.I. Ivanchov, N.V. Pabyrivska, Simultaneous determination of
two coefficients of a parabolic equation in the case of nonlocal and
integral conditions, Ukrainian Math. J. 53 (5) (2001) 674--684.

\bibitem{[15]} M.I. Ismailov, F. Kanca, An inverse coefficient problem for a
parabolic equation in the case of nonlocal boundary and overdetermination
conditions, Math. Methods. Appl. Sci. 34 (6) (2011) 692--702.

\bibitem{[16]} N. B. Kerimov, M. I. Ismailov, An inverse coefficient problem
for the heat equation in the case of nonlocal boundary conditions, J. Math.
Anal. Appl. 396 (2012) 546--554.

\bibitem{[17]} A.I. Prilepko, V.V. Solovev, On the solvability of inverse
boundary value problems for the determination of the coefficient preceding
the lower derivativein a parabolic equation, Differential Equations 23 (1)
(1987) 101--107.

\bibitem{[18]} J.R. Cannon, Y. Lin, An inverse problem of finding a
parameter in a semi-linear heat equation, J. Math. Anal. Appl. 145 (2)
(1990) 470--484.

\bibitem{[19]} D. Lesnic, Identification of the time-dependent perfusion
coefficient in the bio-heat conduction equation, Journal of Inverse and
Ill-posed Problems. 17 (8), (2009) 753--764.

\bibitem{[20]} L.C. Evans, \textit{Partial Differential equation,} AMS,
Providence, 2010.

\bibitem{[21]} S.S. Dragomir, \textit{Some Gronwall Type Inequalities and
Applications,} RGMIA Monographs, Victoria University, Australia, 2002.
\end{thebibliography}
\end{document}